\documentclass[12pt]{amsart}
\usepackage{fullpage}
\usepackage{url}
\usepackage{mathrsfs}

\newtheorem{theorem}{Theorem}
\newtheorem{corollary}[theorem]{Corollary}
\newtheorem{lemma}[theorem]{Lemma}

\begin{document}

\title{Growth rate of binary words avoiding $xxx^R$}

\author{James Currie and Narad Rampersad}

\address{Department of Mathematics and Statistics \\
University of Winnipeg \\
515 Portage Avenue \\
Winnipeg, Manitoba R3B 2E9 (Canada)}

\email{\{j.currie,n.rampersad\}@uwinnipeg.ca}

\thanks{Both authors are supported by NSERC Discovery Grants. The first author also thanks Deutsche Forschungsgemeinschaft, which supported him through its Mercator program.}


\date{\today}

\maketitle
\begin{abstract}
Consider the set of those binary words with no non-empty factors of
the form $xxx^R$.  Du, Mousavi, Schaeffer, and Shallit asked whether
this set of words grows polynomially or exponentially with length. In
this paper, we demonstrate the existence of upper and lower bounds on
the number of such words of length $n$, where each of these bounds is
asymptotically equivalent to a (different) function of the form
$Cn^{\lg n+c}$, where $C$, $c$ are constants and $\lg n$ denotes the
base-$2$ logarithm of $n$.
\end{abstract}

\section{Introduction}
In this paper we study the binary words avoiding the pattern $xxx^R$.
Here the notation $x^R$ denotes the ``reversal'' or ``mirror image''
of $x$.  For example, the word $011\,011\,110$ is an instance of $xxx^R$,
with $x=011$.  The avoidability of patterns with reversals has been
studied before, for instance by Rampersad and Shallit \cite{RS05} and by
Bischoff, Currie, and Nowotka \cite{BCN12, BN11, Cur11}.

The question of whether a given pattern with reversal is avoidable may
initially seem somewhat trivial.  For instance, the pattern $xx^R$ is
avoided by the periodic word $(012)^\omega$ and $xxx^R$, the pattern
studied in this paper, is avoided by the periodic word $(01)^\omega$.
However, looking at the entire class of binary words that avoid
$xxx^R$ reveals that these words have a remarkable structure.

Du, Mousavi, Schaeffer, and Shallit \cite{DMSS14} looked at binary
words avoiding $xxx^R$.  They noted that there are various periodic
words that avoid this pattern and also proved that a certain aperiodic
word studied by Rote \cite{Rot94} and related to the Fibonacci word also
avoids the pattern $xxx^R$.  They posed a variety of conjectures and
open problems concerning binary words avoiding $xxx^R$, notably: Does
the number of such words of length $n$ grow polynomially or
exponentially with $n$?

The growth rate of words avoiding a given pattern over a certain
alphabet is a fundamental problem in combinatorics on words (see the
survey by Shur \cite{Shu11}).
Typically, for families of words defined in terms of the
avoidability of a pattern, this growth is either polynomial or
exponential.  For instance, there are exponentially many ternary words of
length $n$ that avoid the pattern $xx$ and exponentially many binary
words of length $n$ that avoid the pattern $xxx$ \cite{Bra83}.
Similarly, there are exponentially many words over a $4$-letter alphabet
that avoid the pattern $xx$ in the abelian sense \cite{Car98}.  Indeed, the
vast majority of avoidable patterns lead to exponential growth.
Polynomial growth is rather rare:  The two known examples are binary
words avoiding overlaps \cite{RS84} and words over a $4$-letter alphabet
avoiding the pattern $abwbcxaybazac$ \cite{BMT89}.  It was therefore
quite natural for Du et al.\ to suppose that the growth
of binary words avoiding $xxx^R$ was either polynomial or
exponential.  However, we will show that in this case the growth is
intermediate between these two possibilities.  To our knowledge, this
is the first time such a growth rate has been shown in the context of
pattern avoidance.

Our main result is a ``structure theorem'' analogous to the well-known
result of Restivo and Salemi \cite{RS84} concerning binary
overlap-free words.  The existence of such a structure theorem was
conjectured by Shallit (personal communication) but he could not
precisely formulate it.  The result of Restivo and Salemi implies the
polynomial growth of binary overlap-free words.  In our case, the
structure theorem we obtain leads to an upper bound of the form
$Cn^{\lg n+c}$ for binary words avoiding $xxx^R$ (here $\lg n$ denotes
the base-$2$ logarithm of $n$).  We also
are able to establish a lower bound of the same type.  In
Table~\ref{exact_enum} we give an exact enumeration for small values
of $n$.

\begin{table}[h]
\begin{tabular}{||c|c||c|c||c|c||c|c||}
\hline
$n$ & $a_n$ & $n$ & $a_n$ & $n$ & $a_n$ & $n$ & $a_n$ \\
\hline
1 & 2 & 17 & 282 & 33 & 2018 & 49 & 8598 \\
2 & 4 & 18 & 324 & 34 & 2244 & 50 & 9266 \\
3 & 6 & 19 & 372 & 35 & 2490 & 51 & 9964 \\
4 & 10 & 20 & 426 & 36 & 2756 & 52 & 10708 \\
5 & 16 & 21 & 488 & 37 & 3044 & 53 & 11484 \\
6 & 24 & 22 & 556 & 38 & 3354 & 54 & 12300 \\
7 & 34 & 23 & 630 & 39 & 3690 & 55 & 13166 \\
8 & 48 & 24 & 712 & 40 & 4050 & 56 & 14062 \\
9 & 62 & 25 & 804 & 41 & 4438 & 57 & 15000 \\
10 & 80 & 26 & 908 & 42 & 4856 & 58 & 15974 \\
11 & 100 & 27 & 1024 & 43 & 5300 & 59 & 16994 \\
12 & 124 & 28 & 1152 & 44 & 5772 & 60 & 18076 \\
13 & 148 & 29 & 1296 & 45 & 6272 & 61 & 19206 \\
14 & 178 & 30 & 1454 & 46 & 6800 & 62 & 20388 \\
15 & 210 & 31 & 1626 & 47 & 7370 & 63 & 21632 \\
16 & 244 & 32 & 1814 & 48 & 7966 & 64 & 22924 \\
\hline
\end{tabular}
\caption{Number of binary words $a_n$ of length $n$ avoiding
  $xxx^R$}\label{exact_enum}
\end{table}

The sequence $(a_n)_{n \geq 1}$ is sequence A241903 of the On-Line
Encyclopedia of Integer Sequences \cite{oeis}.

\section{Blocks $L$ and $S$}

Define
\[
\mathcal{K} = \{z \in 0\{0,1\}^*1 : z \text{ avoids } xxx^R\}.
\]

Let the transduction $h:\{L,S\}^*\rightarrow\{0,1\}^*$ be defined for a sequence
$u=\prod_{i=0}^n u_i$, $u_i\in\{L,S\}$ by
$$h(u_i)=
\begin{cases}
00100&u_i=S\mbox{ and }i\mbox{ even}\\
11011&u_i=S\mbox{ and }i\mbox{ odd}\\
00100100&u_i=L\mbox{ and }i\mbox{ even}\\
11011011&u_i=L\mbox{ and }i\mbox{ odd}.\\
\end{cases}
$$

Then define 
\[
\mathcal{M} = \{u \in \{S,L\}^* : h(u) \text{ avoids } xxx^R\}.
\]

\begin{theorem}\label{binary to h} Let $z \in \mathcal{K}$.
Then there exists a constant $C$ such that $z$ can be written
\[
z = p h(u) s t
\]
where $|p|,|s| \leq C$, $u \in \mathcal{M}$, and $t \in
(\epsilon+1)(01)^*(\epsilon+1)$.
\end{theorem}
\begin{proof}
 Word $z$ cannot contain 000 or 111 as a factor, so write $z=f(v)$ where $v\in\{ab,ad,cb,cd\}^*$, and
$$f:a\mapsto 0, b\mapsto 1,c\mapsto 00, d\mapsto 11.$$

Write $v=prs$ where $r$ is a maximal string of alternating $a$'s and $b$'s in $v$; thus $r$ lies in $(\epsilon+b)(ab)^*(\epsilon+a)$. If $|s|\ge 2$, then we claim that $|r|=1$ or $|pr|<3$. For suppose that $|r|\ge 2$, $|pr|\ge 3$ and $|s|\ge 2$. Let $s_1$, $s_2$  be the first two letters of $s$. Then $s_1$ must be $c$ or $d$; otherwise, $rs_1$ is an alternating string of $a$'s and $b$'s that is longer than $r$. Suppose $s_1=c$. (The other case is similar.) Since $|r|\ge 2$ and $|pr|\ge 3$, we conclude that $prs_1s_2$ has $yabcs_2$ as a suffix, some $y\in\{b,d\}$. But then $z$ contains a factor $f(yabcs_2)$, which has a factor  $1f(abc)1=101001=xxx^R$, where $x=10$. This is impossible.

If $ab$ or $ba$ is a factor of $v$, we can write $v=prs$ as above, with $|r|\ge 2$. This implies that $|s|\le 1$ or $|pr|\le 2$. If $|pr|\le 2$, then $p=\epsilon$, $|r|=2$, since $|r|\ge 2$; in this case $pr=ab$. If $s\le 1$, then, since $z$ ends in 1, either $s=\epsilon$ or $s=d$. In the first case, $ab$ is a suffix of $v$; in the second $ad$ is a suffix. It follows that every instance of $ab$ or $ba$ in $v$ either occurs in a prefix of length 2, or in a suffix of the form $(\epsilon +b)(ab)^*(\epsilon+ad)$.
The given suffix maps under $f$ to a suffix  $t \in(\epsilon+1)(01)^*(\epsilon+1)$ of $z$. We therefore can write $z=p_1z_1t$ such that $|p_1|\le 2$, and $z_1=f(v_1)$, for some $v_1\in\{ad,cb,cd\}^*$ where $ba$ is not a factor of $v_1$.

Write $v_1=prs$ where $r$ is a maximal string of alternating $c$'s and $d$'s in $v_1$. First of all, note that $|r|<7$; we check that $f(cdcdcdc)$ contains $xxx^R$ with $x=0d0$, and, symmetrically, $f(dcdcdcd)$ contains $xxx^R$ with $x=1c1$. We claim that $|r|<3$ or $|pr|<7$. For otherwise, suppose that $|r|\ge 3$, and $|p'r|= 7$, where $p'$ is a suffix of $p$. Assume that the first letter of $r$ is $c$. (The other case is similar.) Since $|r|<7$, $p'\ne\epsilon$. Since $r$ is maximal, the last letter of $p'$ is a $b$. If $|p'|=1$, then $f(p'r)=f(bcdcdcd)$, which contains $xxx^R$ with $x=1c1$; this is impossible. If $|p'|\ge 2$, then $cb$ is a suffix of $p'$ (since $ab$ is not a factor of $v_1$.) However, then $p'r$ contains the factor $cbcdc$, and $f(cbcdc)=001001100=xxx^R$, where $x=001$, so this is also impossible. It follows that every instance of $cdc$ or $dcd$ in $v_1$ occurs in a prefix of $v_1$ of length $6$. Removing a prefix $p'$ of length at most 7 from $v_1$ then gives a suffix $v_2$, such that the first letter of $v_2$ is $a$ or $c$, and neither of $cdc$ and $dcd$ is a factor of $v_2$. We can thus write $z=p_2z_2t$ where $z_2=f(v_2)$, $v_2\in\{ad,cb,cd\}^*$, words $ba$, $cdc$, $dcd$ are not factors of $v_2$, and $|p_2|\le |p_1|+|f(p')|\le 2 + 2(7)-1=15$. (Here, at most 6 letters of $p'$ can be $c$ or $d$, since $cdcdcdc$ and $dcdcdcd$ lead to instances of $xxx^R$.)

Suppose that $v'$ is any factor of $v_2$ of length 8. We claim that $v'$ contains one of $cd$ or $dc$ as a factor. Since $v'\notin\{a,b\}^*$, one of $c$ and $d$ is a factor of $v'$. Suppose then that $c$ is a factor of $v'$. (The other case is similar.) Suppose that neither of $cd$ nor $dc$ is a factor of $v'$. It follows that $v'$ is $bcbcbcbc$ or $cbcbcbcb$; each of these contains $cbcbcbc$, and $f(cbcbcbc)$ contains $010010010=xxx^R$ where $x=010$.

We may thus write $v_2=p'\left(\prod_{i=0}^n a_i\right)s'$, with $n\ge -1$, $|p'|,|s'|\le 7$, such that each $a_i$ begins and ends with $c$ or $d$, and neither of $cd$ or $dc$ is a factor of any $a_i$. By $n=-1$ we allow the possibility that the product term is empty. As a convention, we write the product as empty if $|v_2|_{cd}+|v_2|_{dc}\le 1$; for $i\ge 0$, then the last letter of $p'$ and the first letter of $s'$ are in $\{c,d\}$.
Suppose $n\ge 0$. Consider $a_i$, $i\ge 0$. Without loss of generality, let $a_i$ begin with $c$. The letter preceding $a_i$ is either the last letter of $a_{i-1}$, or the last letter of $p'$, and must be a $d$. We cannot have $|a_i|=1$, which would force $a_i=c$; word $a_i$ is then followed by the first letter of $a_{i+1}$ or of $s'$, which must be $d$. Then $dcd$ is a factor of $v_2$, which is impossible. Thus $|a_i|\ge 2$. Since $cd$ is not a factor of $a_i$, $a_i$ begins with $cb$. Since $a_i$ ends with $c$ or $d$ (not in $b$), $a_i\ne cb$, so that $|a_i|\ge 3$. Since $ba$ is not a factor of $v_2$, $a_i$ therefore begins with $cbc$. If $a_i\ne cbc$, then, since $cd$ is not a factor of $a_i$, word $a_i$ begins with $cbcb$, and arguing as previously, with $cbcbc$. If $cbcbc$ is a proper prefix of $a_i$, then $a_i$ begins with $cbcbcb$. However, $f(cb)^30$ contains an instance of $xxx^R$, so this is impossible: If $a_i$ begins with $c$, then $a_i\in\{cbc,cbcbc\}$. By the same reasoning, if $a_i$ begins with $d$, then $a_i\in\{dad,dadad\}$.

Let $v_3=(p')^{-1}v_2(s')^{-1}=\prod_{i=0}^n a_i$. Deleting up to the first 5 letters, if necessary, we assume that $a_0\in\{cbc,cbcbc\}$ (i.e., if $a_0$ begins with $dad$ or $dadad$, then delete these letters.) Then 
$z=p_3z_3s_3t$ where $z_3=f(v_3)$, $|p_3|\le |f(p')|+|p_2|+5\le 2(4)+3+15+5=31$, $|s_3|=|f(s')|\le 2(4)+3=11$. 
Here we use the fact that at most 4 of the letters of $p'$ or $s'$ can be in $\{c,d\}$; otherwise the pigeonhole principle would force an occurrence of $cd$ or $dc$ in one of these.

We can write $v_3$ in the form $g(u)$ where $u\in\{L,S\}^*$. Here write $u=\prod_{i=0}^m u_i$, each $u_i\in\{L,S\}$, and let $g$ be the transducer $$g(u_i)=
\begin{cases}
cbc&u_i=S\mbox{ and }i\mbox{ even}\\
dad&u_i=S\mbox{ and }i\mbox{ odd}\\
cbcbc&u_i=L\mbox{ and }i\mbox{ even}\\
dadad&u_i=L\mbox{ and }i\mbox{ odd}.\\
\end{cases}
$$
Thus $z_3$ has the form
$h(u)$ where $h$ is the transducer $$h(u_i)=
\begin{cases}
00100&u_i=S\mbox{ and }i\mbox{ even}\\
11011&u_i=S\mbox{ and }i\mbox{ odd}\\
00100100&u_i=L\mbox{ and }i\mbox{ even}\\
11011011&u_i=L\mbox{ and }i\mbox{ odd}.\\
\end{cases}
$$
We have thus proved the theorem with $C=\max(31,11)=31.$
\end{proof}

To study the growth rate of $\mathcal{K}$, it thus suffices to study the growth rate of $\mathcal{M}$.

The transducer $h$ is sensitive to the index of a word modulo 2; thus, suppose $r$, $s\in\{L,S\}^*$ and $r$ is a suffix of $s$. If $|r|$ and $|s|$ have the same parity, then $h(r)$ is a suffix of $h(s)$. However, if $|r|$ and $|s|$ have opposite parity, then $\overline{h(r)}$ is a suffix of $h(s)$.  (Here the overline indicates binary complementation.)

\section{Suitable pairs of words}

Let $\mathcal{S},\mathcal{L}\in\{S,L\}^*$. Say that the pair $\langle\mathcal{S},\mathcal{L}\rangle$ is {\bf suitable} if 

\begin{enumerate}
\item $|\mathcal{S}|$, $|\mathcal{L}|$ are odd.
\item There exist non-empty $\ell,\mu,p\in\{L,S\}^*$ such that
\begin{enumerate}
\item  $h(\mathcal{L})=\ell\ell^R$
\item  $h(\mathcal{S})=\ell\mu=\mu^R\ell^R$
\item $h(\mathcal{L})=\ell\mu\overline{\mu^R}p$

\end{enumerate}
\end{enumerate}

We see that $\langle S,L\rangle$ is suitable;
specifically, we could choose $\mu=0$, $\ell=0010$, $p=00$.\vspace{.1in}

Since $|\mathcal{S}|$, $|\mathcal{L}|$ are odd, the transducer $h$ is sensitive to the index of a word modulo 2, where lengths (and indices) are measured in terms of $\mathcal{S}$ and $\mathcal{L}$; i.e., if we use length function $||w||=|w|_\mathcal{S}+|w|_\mathcal{L}$; thus, suppose $r$, $s\in\{\mathcal{S},\mathcal{L}\}^*$ and $r$ is a suffix of $s$. If $||r||$ and $||s||$ have the same parity, then $h(r)$ is a suffix of $h(s)$. However, if $||r||$ and $||s||$ have opposite parity, then $\overline{h(r)}$ is a suffix of $h(s)$.

\begin{lemma}\label{properties of h}
Let $\mathcal{S},\mathcal{L}\in\{S,L\}^*$. Suppose that $\langle\mathcal{S},\mathcal{L}\rangle$ is suitable.
\begin{enumerate}
\item Word $h(\mathcal{L})p^{-1}$ is a prefix of $h(\mathcal{S}\mathcal{S})$.
\item Word $h(\mathcal{S})$ is  both a prefix and suffix of $h(\mathcal{L})$.
\end{enumerate}
\end{lemma}

\begin{proof}
The first of these properties is immediate from property 2(c) of the definition of suitability. For the second, we see that
$h(\mathcal{L})=\ell\mu\overline{\mu^R}p=\mu^R\ell^R\overline{\mu^R}p=p^R\overline{\mu}\ell\mu$.
\end{proof}

Now suppose that $\mathcal{S}$ and $\mathcal{L}$ are fixed and $\langle\mathcal{S},\mathcal{L}\rangle$ is suitable. 
Define morphism $\Phi:\{\mathcal{S},\mathcal{L}\}^*\rightarrow \{\mathcal{S},\mathcal{L}\}^*$ by
$\Phi(\mathcal{S})=\mathcal{S}\mathcal{L}$, $\Phi(\mathcal{L})=\mathcal{S}\mathcal{L}\mathcal{L}$. 

Morphism $\Phi$ is conjugate to the square of the Fibonacci morphism
$D$, where $D(\mathcal{L})=\mathcal{LS}$,
$D(\mathcal{S})=\mathcal{L}$; namely,
$\Phi=\mathcal{L}^{-1}D^2\mathcal{L}$. This implies that
$||\Phi^k(\mathcal{S})||=\mathcal{F}_{2k}$,
$||\Phi^k(\mathcal{L})||=\mathcal{F}_{2k+1}$, where $\mathcal{F}_k$ is
the $k$th Fibonacci number, counting from
$\mathcal{F}_0=\mathcal{F}_1=1$ (we choose this indexing of the
Fibonacci numbers for convenience: in particular, so that
$||\Phi^0(\mathcal{S})||=\mathcal{F}_{0}=1$).

\begin{lemma}\label{properties of h and phi}
Let $\beta\in\{\mathcal{S},\mathcal{L}\}^*$. Then
\begin{enumerate}
\item $h(\Phi(\mathcal{S}\beta))$ is a prefix of $h(\Phi(\mathcal{L}\beta))$ and  $h(\Phi^2(\mathcal{S}\beta))$ is a prefix of $h(\Phi^2(\mathcal{L}\beta))$. 
\item $\overline{h(\Phi(\mathcal{S}\beta))}$ is a suffix of $h(\Phi(\mathcal{L}\beta)).$
\item $\overline{h(\Phi^2(\mathcal{S}\beta))}$ is a suffix of $h(\Phi^2(\mathcal{L}\beta)).$
\item $h(\Phi(\mathcal{L}))p^{-1}$ is a prefix of $h(\Phi(\mathcal{S}\mathcal{S}))$.
\item $h(\Phi^2(\mathcal{L}))(\overline{p})^{-1}$ is a prefix of $h(\Phi^2(\mathcal{S}\mathcal{S}))$.
\end{enumerate}

\end{lemma}

\begin{proof}
Since $\Phi(\mathcal{S})$ is a prefix of $\Phi(\mathcal{L})$, $\Phi(\mathcal{S}\beta)$ is a prefix of $\Phi(\mathcal{L}\beta)$, so that $h(\Phi(\mathcal{S}\beta))$ is a prefix of $h(\Phi(\mathcal{L}\beta))$. Similarly,
 $h(\Phi^2(\mathcal{S}\beta))$ is a prefix of $h(\Phi^2(\mathcal{L}\beta))$, establishing (1).

Since $\mathcal{S}$ is a suffix of $\mathcal{L}$, we see that $\Phi(\mathcal{S})$ is a suffix of 
$\Phi(\mathcal{L})$. Because
 $|\Phi(\mathcal{L})|$ is odd, while $|\Phi(\mathcal{S})|$ is even,
it follows that $\overline{h(\Phi(\mathcal{S}))}$ is a suffix of $h(\Phi(\mathcal{L})).$ More generally,
if $\beta\in\{\mathcal{S},\mathcal{L}\}^*$, $\overline{h(\Phi(\mathcal{S}\beta))}$ is a suffix of $h(\Phi(\mathcal{L}\beta))$, establishing (2). The proof of (3) is similar.

For (4),
$h(\Phi(\mathcal{L}))p^{-1}=h(\mathcal{S}\mathcal{L}\mathcal{L})p^{-1}
=h(\mathcal{S}\mathcal{L})h(\mathcal{L})p^{-1}$, which is a prefix of
$h(\mathcal{S}\mathcal{L})h(\mathcal{S}\mathcal{S}),$ which is in turn a prefix of
$h(\mathcal{S}\mathcal{L})h(\mathcal{S}\mathcal{L})=h(\Phi(\mathcal{SS}))$.

For (5),
$h(\Phi^2(\mathcal{L}))(\overline{p})^{-1}=h(\Phi(\mathcal{S}\mathcal{L})\Phi(\mathcal{L}))(\overline{p})^{-1}
=h(\Phi(\mathcal{S}\mathcal{L}))\overline{h(\Phi(\mathcal{L}))p^{-1}}$ (since $|\Phi(\mathcal{S}\mathcal{L})|$ is odd), which is a prefix of $
h(\Phi(\mathcal{S}\mathcal{L}))\overline{h(\Phi(\mathcal{SS}))}=h(\Phi(\mathcal{S}\mathcal{L})\Phi(\mathcal{SS}))$, which is in turn a prefix of
$h(\Phi(\mathcal{S}\mathcal{L}\mathcal{SSL})),=h(\Phi^2(\mathcal{SS}))$.

\end{proof}

Define the set $\mathcal{B}\subseteq \{ \mathcal{S},\mathcal{L}\}^*$:

\begin{eqnarray*}
\mathcal{B}&=&(\mathcal{S}+\mathcal{L})\mathcal{S}\mathcal{S}\mathcal{S}\mathcal{L}(\mathcal{L}+\mathcal{S}\mathcal{S}+\mathcal{S}\mathcal{L})
\cup 
\mathcal{L}\mathcal{S}\mathcal{S}\mathcal{L}(\mathcal{L}+\mathcal{S}\mathcal{S}+\mathcal{S}\mathcal{L})
\cup(\mathcal{S}+\mathcal{L})\mathcal{L}\mathcal{L}\mathcal{L}\mathcal{L}\mathcal{L}(\mathcal{S}+\mathcal{L})\\
&&
\cup(\mathcal{S}+\mathcal{L})\mathcal{L}\mathcal{S}\mathcal{L}\mathcal{L}\mathcal{L}(\mathcal{S}+\mathcal{L})\\&&
\cup\Phi((\mathcal{S}+\mathcal{L})\mathcal{S}\mathcal{S}(\mathcal{S}+\mathcal{L}))
\cup\Phi((\mathcal{S}+\mathcal{L})
\mathcal{L}\mathcal{L}\mathcal{L}\mathcal{S}\mathcal{L}(\mathcal{L}+\mathcal{S}\mathcal{S}+\mathcal{S}\mathcal{L}))\\
&&
\cup\Phi^2(\mathcal{L}\mathcal{L}\mathcal{L}(\mathcal{S}+\mathcal{L}))
\cup\Phi^2((\mathcal{S}+\mathcal{L})\mathcal{L}\mathcal{S}\mathcal{S}(\mathcal{S}+\mathcal{L}))
\cup\Phi^2((\mathcal{S}+\mathcal{L})\mathcal{S}\mathcal{S}\mathcal{S}\mathcal{S}\mathcal{S}(\mathcal{S}+\mathcal{L}))
\end{eqnarray*}

\begin{lemma}\label{B-lemma}
Let $u\in\mathcal{M}$. Then no word of $\mathcal{B}$ is a factor of $u$.
\end{lemma}
\begin{proof}
It suffices to show that for each word $b\in\mathcal{B}$, $h(b)$ contains a non-empty factor $xxx^R$.
$\mathcal{B}$ is written as a union, and we make cases based on which piece of the union $b$ belongs:
\vspace{.1in}

\noindent $b\in (\mathcal{S}+\mathcal{L})\mathcal{S}\mathcal{S}\mathcal{S}\mathcal{L}(\mathcal{L}+\mathcal{S}\mathcal{S}+\mathcal{S}\mathcal{L})$: In this case, it suffices to show that
 $h(\mathcal{S}\mathcal{S}\mathcal{S}\mathcal{S}\mathcal{L}\mathcal{L})(\bar{p})^{-1}$ contains a non-empty factor $xxx^R$, because of the results of Lemma~\ref{properties of h}. In particular, $h(\mathcal{S}\mathcal{S}\mathcal{S}\mathcal{S}\mathcal{L}\mathcal{L})(\bar{p})^{-1}$ is a suffix of
$h(\mathcal{L}\mathcal{S}\mathcal{S}\mathcal{S}\mathcal{L}\mathcal{L})(\bar{p})^{-1}$, which is a prefix of 
$h(\mathcal{L}\mathcal{S}\mathcal{S}\mathcal{S}\mathcal{L}\mathcal{S}\mathcal{S})$, which is a prefix of 
$h(\mathcal{L}\mathcal{S}\mathcal{S}\mathcal{S}\mathcal{L}\mathcal{S}\mathcal{L})$. Again,
$h(\mathcal{S}\mathcal{S}\mathcal{S}\mathcal{S}\mathcal{L}\mathcal{L})(\bar{p})^{-1}$ is a prefix of 
$h(\mathcal{S}\mathcal{S}\mathcal{S}\mathcal{S}\mathcal{L}\mathcal{S}\mathcal{S})$, which is a prefix of 
$h(\mathcal{S}\mathcal{S}\mathcal{S}\mathcal{S}\mathcal{L}\mathcal{S}\mathcal{L})$. Now
\begin{eqnarray*}
&&h(\mathcal{S}\mathcal{S}\mathcal{S}\mathcal{S}\mathcal{L}\mathcal{L})(\bar{p})^{-1}\\
&=&(\ell\mu)(\overline{\mu^R\ell^R})(\ell\mu)(\overline{\mu^R\ell^R})(\ell\ell^R)(\overline{\ell\mu\overline{\mu^R}})\\
&=&\ell\mu\overline{\mu^R\ell^R}\ell\mu\overline{\mu^R\ell^R}\ell\ell^R\overline{\ell\mu\overline{\mu^R}}\\
&=&\ell\hspace{.1in}\mu\overline{\mu^R\ell^R}\ell\hspace{.1in}\mu\overline{\mu^R\ell^R}\ell
\hspace{.1in}\ell^R\overline{\ell\mu}\mu^R\hspace{.1in}
\end{eqnarray*}
which contains an instance of $xxx^R$ with $x=\mu\overline{\mu^R\ell^R}\ell$. \vspace{.1in}

\noindent $b\in \mathcal{L}\mathcal{S}\mathcal{S}\mathcal{L}(\mathcal{L}+\mathcal{S}\mathcal{S}+\mathcal{S}\mathcal{L})
$: In this case, it suffices to show that
 $h(\mathcal{L}\mathcal{S}\mathcal{S}\mathcal{L}\mathcal{L})p^{-1}$ contains a non-empty factor $xxx^R$, because of the results of Lemma~\ref{properties of h}. But
\begin{eqnarray*}
&&h(\mathcal{L}\mathcal{S}\mathcal{S}\mathcal{L}\mathcal{L})p^{-1}\\
&=&(p\overline{\mu}\mu^R\ell^R)(\overline{\ell\mu})(\mu^R\ell^R)(\overline{\ell\ell^R})(\ell\mu\overline{\mu^R})\\
&=&p\overline{\mu}\mu^R\ell^R\overline{\ell\mu}\mu^R\ell^R\overline{\ell\ell^R}\ell\mu\overline{\mu^R}\\
&=&p\hspace{.1in}\overline{\mu}\mu^R\ell^R\overline{\ell}\hspace{.1in}
\overline{\mu}\mu^R\ell^R\overline{\ell}\hspace{.1in}\overline{\ell^R}\ell\mu\overline{\mu^R}\hspace{.1in}
\end{eqnarray*}
which contains the instance $xxx^R$ with $x=\overline{\mu}\mu^R\ell^R\overline{\ell}$.\vspace{.1in}

\noindent $b\in (\mathcal{S}+\mathcal{L)})\mathcal{L}^5 (\mathcal{S}+\mathcal{L)})$: In this case, it suffices to show that
 $h(\mathcal{S}\mathcal{L}^5\mathcal{S})$ contains a non-empty factor $xxx^R$, because of the results of Lemma~\ref{properties of h}. But
\begin{eqnarray*}
&&h(\mathcal{S}\mathcal{L}^5\mathcal{S})\\&=&(\mu^R\ell^R)(\overline{\ell\ell^R})(\ell\ell^R)(\overline{\ell\ell^R})(\ell\ell^R)(\overline{\ell\ell^R})(\ell\mu)\\
&=&\mu^R\ell^R\overline{\ell\ell^R}\ell\ell^R\overline{\ell\ell^R}\ell\ell^R\overline{\ell\ell^R}\ell\mu\\
&=&\mu^R\hspace{.1in}\ell^R\overline{\ell\ell^R}\ell\hspace{.1in}\ell^R\overline{\ell\ell^R}\ell\hspace{.1in}\ell^R\overline{\ell\ell^R}\ell\hspace{.1in}\mu\\
\end{eqnarray*}
which contains the instance $xxx^R$ with $x=\ell^R\overline{\ell\ell^R}\ell$.\vspace{.1in} 

$b\in(\mathcal{S}+\mathcal{L})\mathcal{L}\mathcal{S}\mathcal{L}\mathcal{L}\mathcal{L}(\mathcal{S}+\mathcal{L}):$ In this case, it suffices to show that
 $h(\mathcal{S}\mathcal{L}\mathcal{S}\mathcal{L}\mathcal{L}\mathcal{L}\mathcal{S})$ contains a non-empty factor $xxx^R$, because of the results of Lemma~\ref{properties of h}. Here

\begin{eqnarray*}
&&h(\mathcal{S}\mathcal{L}\mathcal{S}\mathcal{L}\mathcal{L}\mathcal{L}\mathcal{S})\\
&=&(\ell\mu)(\overline{\ell\ell^R})(\ell\mu)(\overline{\ell\ell^R})(\ell\ell^R)(\overline{\ell\ell^R})(\mu^R\ell^R)\\
&=&\ell\mu\overline{\ell\ell^R}\ell\mu\overline{\ell\ell^R}\ell\ell^R\overline{\ell\ell^R}\mu^R\ell^R\\
&=&\ell\hspace{.1in}\mu\overline{\ell\ell^R}\ell\hspace{.1in}\mu\overline{\ell\ell^R}\ell\hspace{.1in}\ell^R\overline{\ell\ell^R}\mu^R\hspace{.1in}\ell^R\\
\end{eqnarray*}
which contains the instance $xxx^R$ with $x=\mu\overline{\ell\ell^R}\ell$. \vspace{.1in} 

$b\in\Phi((\mathcal{S}+\mathcal{L})\mathcal{S}\mathcal{S}(\mathcal{S}+\mathcal{L})):$ In this case, it suffices to show that
 $h(\Phi(\mathcal{S}\mathcal{S}\mathcal{S}\mathcal{S}))$ contains a non-empty factor $xxx^R$, because of the results of Lemma~\ref{properties of h and phi}. In particular, $h(\Phi(\mathcal{S}\mathcal{S}\mathcal{S}\mathcal{S}))$ is a prefix of $h(\Phi(\mathcal{S}\mathcal{S}\mathcal{S}\mathcal{L}))$, 
$\overline{h(\Phi(\mathcal{S}\mathcal{S}\mathcal{S}\mathcal{S}))}$
is a suffix of $h(\Phi(\mathcal{L}\mathcal{S}\mathcal{S}\mathcal{S}))$,
and $\overline{h(\Phi(\mathcal{S}\mathcal{S}\mathcal{S}\mathcal{L}))}$
is a suffix of $h(\Phi(\mathcal{L}\mathcal{S}\mathcal{S}\mathcal{L}))$. However,
\begin{eqnarray*}
&&h(\Phi(\mathcal{S}\mathcal{S}\mathcal{S}\mathcal{S}))\\
&=&(\mu^R\ell^R\overline{\ell\ell^R})(\mu^R\ell^R\overline{\ell\ell^R})(\mu^R\ell^R\overline{\ell\ell^R}(\mu^R\ell^R\overline{\ell\ell^R})\\
&=&\mu^R\ell^R\overline{\ell}\hspace{.1in}\overline{\ell^R}\mu^R\ell^R\overline{\ell}\hspace{.1in}
\overline{\ell^R}\mu^R\ell^R\overline{\ell}\hspace{.1in}\overline{\ell^R}\mu^R\ell^R
\overline{\ell}\hspace{.1in}\overline{\ell^R}\end{eqnarray*}
containing an instance of $xxx^R$, with $x=\overline{\ell^R}\mu^R\ell^R\overline{\ell}$. 
\vspace{.1in} 

$b\in\Phi((\mathcal{S}+\mathcal{L})
\mathcal{L}\mathcal{L}\mathcal{L}\mathcal{S}\mathcal{L}(\mathcal{L}+\mathcal{S}\mathcal{S}+\mathcal{S}\mathcal{L})):$ In this case, it suffices to show that
 $h(\Phi(\mathcal{S}\mathcal{L}\mathcal{L}\mathcal{L}\mathcal{S}\mathcal{L}\mathcal{L}))p^{-1}$ contains a non-empty factor $xxx^R$, because of the results of Lemma~\ref{properties of h and phi}. But

\begin{eqnarray*}
&&h(\Phi(\mathcal{S}\mathcal{L}\mathcal{L}\mathcal{L}\mathcal{S}\mathcal{L}\mathcal{L}))p^{-1}\\
&=&(\mu^R\ell^R\overline{\ell\ell^R})
(\mu^R\ell^R\overline{\ell\ell^R}\ell\ell^R)
(\overline{\mu^R\ell^R}\ell\ell^R\overline{\ell\ell^R})
(\mu^R\ell^R\overline{\ell\ell^R}\ell\ell^R)
(\overline{\mu^R\ell^R}\ell\ell^R)
(\overline{\ell\mu}\ell\ell^R\overline{\ell\ell^R})
(\ell\mu\overline{\ell\ell^R}
\ell\mu\overline{\mu^R})\\
&=&\mu^R\hspace{.1in}\ell^R\overline{\ell\ell^R}
\mu^R\ell^R\overline{\ell\ell^R}\ell\ell^R
\overline{\mu^R\ell^R}\ell
\hspace{.1in}
\ell^R\overline{\ell\ell^R}
\mu^R\ell^R\overline{\ell\ell^R}\ell\ell^R
\overline{\mu^R\ell^R}\ell
\hspace{.1in}
\ell^R
\overline{\ell\mu}\ell\ell^R\overline{\ell\ell^R}
\ell\mu\overline{\ell\ell^R}
\ell\hspace{.1in}\mu\overline{\mu^R}
\end{eqnarray*}
 an instance of $xxx^R$ with $x=\ell^R\overline{\ell\ell^R}
\mu^R\ell^R\overline{\ell\ell^R}\ell\ell^R
\overline{\mu^R\ell^R}\ell
$.\vspace{.1in}

$b\in\Phi^2(\mathcal{L}\mathcal{L}\mathcal{L}(\mathcal{S}+\mathcal{L})):$ In this case, it suffices to show that
 $h(\Phi^2(\mathcal{L}\mathcal{L}\mathcal{L}\mathcal{S}))$ contains a non-empty factor $xxx^R$, because of the results of Lemma~\ref{properties of h and phi}. But

\begin{eqnarray*}
&&h(\Phi^2(\mathcal{L}\mathcal{L}\mathcal{L}\mathcal{S}))\\
&=&(\ell\mu\overline{\ell\ell^R}
\ell\mu\overline{\ell\ell^R} \ell\ell^R
\overline{\ell\mu} \ell\ell^R\overline{\ell\ell^R})
( \ell\mu\overline{\ell\ell^R}
\ell\mu\overline{\ell\ell^R} \ell\ell^R
\overline{\ell\mu} \ell\ell^R\overline{\ell\ell^R})
(\ell\mu\overline{\ell\ell^R}
\mu^R\ell^R\overline{\ell\ell^R} \ell\ell^R
\overline{\mu^R\ell^R} \ell\ell^R\overline{\ell\ell^R})
 \\&&
(\mu^R\ell^R\overline{\ell\ell^R}\ell\mu\overline{\ell\ell^R} \ell\ell^R)
\\
&=&\ell\mu\overline{\ell}\cdot\overline{\ell^R}
\ell\mu\overline{\ell\ell^R} \ell\ell^R
\overline{\ell\mu} \ell\ell^R\overline{\ell\ell^R}
\ell\mu\overline{\ell}\cdot\overline{\ell^R}
\ell\mu\overline{\ell\ell^R} \ell\ell^R
\overline{\ell\mu} \ell\ell^R\overline{\ell\ell^R}
\ell\mu\overline{\ell}\cdot\overline{\ell^R}
\mu^R\ell^R\overline{\ell\ell^R} \ell\ell^R
\overline{\mu^R\ell^R} \ell\ell^R\overline{\ell\ell^R}
 \mu^R\ell^R\overline{\ell}\\
&&\cdot\overline{\ell^R}\ell\mu\overline{\ell\ell^R} \ell\ell^R
\end{eqnarray*}
containing an instance of $xxx^R$, with $x=
\overline{\ell^R}
\ell\mu\overline{\ell\ell^R} \ell\ell^R
\overline{\ell\mu} \ell\ell^R\overline{\ell\ell^R}
\ell\mu\overline{\ell}$. 

$b\in\Phi^2((\mathcal{S}+\mathcal{L})\mathcal{L}\mathcal{S}\mathcal{S}(\mathcal{S}+\mathcal{L})):$ In this case, it suffices to show that
 $h(\Phi^2(\mathcal{S}\mathcal{L}\mathcal{S}\mathcal{S}\mathcal{S}))$ contains a non-empty factor $xxx^R$, because of the results of Lemma~\ref{properties of h and phi}. Now

\begin{eqnarray*}
&&h(\Phi^2(\mathcal{S}\mathcal{L}\mathcal{S}\mathcal{S}\mathcal{S}))\\
&=&(\ell\mu\overline{\ell\ell^R}
\ell\mu\overline{\ell}\overline{\ell^R}\ell\ell^R)
(\overline{\ell\mu}\ell\ell^R\overline{\ell\mu} \ell\ell^R\overline{\ell\ell^R}\ell\mu\overline{\ell}
\overline{\ell^R} \ell\ell^R)
(\overline{\ell\mu}\ell\ell^R
\overline{\ell\mu}\ell\ell^R\overline{\ell\ell^R} )
(\ell\mu\overline{\ell}\overline{\ell^R}
\mu^R\ell^R\overline{\ell\ell^R} \ell\ell^R)
\\
&&(\overline{\mu^R\ell^R} \ell\ell^R
\overline{\mu^R\ell^R} \ell\ell^R\overline{\ell}\overline{\ell^R})
\\
&=&\ell\mu\overline{\ell\ell^R} \ell\mu\overline{\ell}
\cdot\overline{\ell^R} \ell\ell^R\overline{\ell\mu} \ell\ell^R\overline{\ell\mu} \ell\ell^R\overline{\ell\ell^R} \ell\mu\overline{\ell}
\cdot\overline{\ell^R} \ell\ell^R\overline{\ell\mu}
\ell\ell^R\overline{\ell\mu} \ell\ell^R\overline{\ell\ell^R} 
\ell\mu\overline{\ell}
\\&&
\cdot\overline{\ell^R} \mu^R\ell^R
\overline{\ell\ell^R} \ell\ell^R\overline{\mu^R\ell^R} \ell\ell^R
\overline{\mu^R\ell^R} \ell\ell^R\overline{\ell}
\cdot\overline{\ell^R}
\end{eqnarray*}
containing an instance of $xxx^R$, with $x=
\overline{\ell^R} \ell\ell^R\overline{\ell\mu} \ell\ell^R\overline{\ell\mu} \ell\ell^R\overline{\ell\ell^R} \ell\mu\overline{\ell}.$\vspace{.1in}

$b\in\Phi^2((\mathcal{S}+\mathcal{L})\mathcal{S}\mathcal{S}\mathcal{S}\mathcal{S}\mathcal{S}(\mathcal{S}+\mathcal{L})):$ In this case, it suffices to show that
 $h(\Phi^2(\mathcal{S}^7))$ contains a non-empty factor $xxx^R$, because of the results of Lemma~\ref{properties of h and phi}. Finally,

\begin{eqnarray*}
&&h(\Phi^2(\mathcal{S}^7))\\
&=&
(\ell\mu\overline{\ell\ell^R}\ell\mu\overline{\ell\ell^R}\ell\ell^R)
(\overline{\ell\mu}\ell\ell^R\overline{\mu^R\ell^R}\ell\ell^R\overline{\ell\ell^R})
(\ell\mu\overline{\ell\ell^R}\ell\mu\overline{\ell\ell^R}\ell\ell^R)
(\overline{\ell\mu}\ell\ell^R\overline{\mu^R\ell^R}\ell\ell^R\overline{\ell\ell^R})
\\&&
(\ell\mu\overline{\ell\ell^R}\ell\mu\overline{\ell\ell^R}\ell\ell^R)
(\overline{\ell\mu}\ell\ell^R\overline{\mu^R\ell^R}\ell\ell^R\overline{\ell\ell^R})
(\ell\mu\overline{\ell\ell^R}\ell\mu\overline{\ell\ell^R}\ell\ell^R)
\\ &=&\ell\mu\overline{\ell}
\cdot
\overline{\ell^R}\ell\mu
\overline{\ell\ell^R}
\ell\ell^R
\overline{\ell\mu}
\ell
\ell^R
\overline{\mu^R\ell^R}
\ell\ell^R 
\overline{\ell\ell^R}
\mu^R\ell^R\overline{\ell}
\cdot
\overline{\ell^R}\ell\mu
\overline{\ell\ell^R}
\ell\ell^R
\overline{\ell\mu}
\ell
\ell^R
\overline{\mu^R\ell^R}
\ell\ell^R
\overline{\ell\ell^R}
\mu^R\ell^R\overline{\ell}
\\
&&\cdot
\overline{\ell^R}\ell\mu
\overline{\ell\ell^R}\ell\ell^R
\overline{\ell\mu}
\ell\ell^R
\overline{\mu^R\ell^R}
\ell\ell^R 
\overline{\ell\ell^R}
\mu^R\ell^R\overline{\ell}
\cdot
\overline{\ell^R}\mu^R\ell^R\overline{\ell\ell^R}\ell\ell^R
\end{eqnarray*}

containing an instance of $xxx^R$, with $x=
\overline{\ell^R}\ell\mu
\overline{\ell\ell^R}
\ell\ell^R
\overline{\ell\mu}
\ell
\ell^R
\overline{\mu^R\ell^R}
\ell\ell^R
\overline{\ell\ell^R}
\mu^R\ell^R\overline{\ell}.$\end{proof}

\section{Parsing words of $\mathcal{M}$ using $\Phi$}

\begin{lemma}\label{SL}Let $y\in \{\mathcal{L},\mathcal{S}\}^*\cap\mathcal{M}$. Then  $y$ can be written
\[
y = p_1 \Phi(y_1) s_1t_1,
\]
where $|p_1|,|s_1| \leq 9$, $y_1 \in \{\mathcal{L},\mathcal{S}\}^*$, and $t_1 \in
(\epsilon+\mathcal{S}+\mathcal{S}^2+\mathcal{S}^3)\mathcal{L}\mathcal{S}^*+
\mathcal{S}^*(\epsilon+\mathcal{L}+\mathcal{L}\mathcal{S})$. (Here all lengths are as words of $ \{\mathcal{L},\mathcal{S}\}^*$; thus, for example $|p_1|=|p_1|_\mathcal{L}+|p_1|_\mathcal{S}$.)
\end{lemma}
\begin{proof}
Suppose that $|y|_\mathcal{L}=n$. If $n=0$, the lemma is true, letting $t_1=y$. If $n=1$, write $y=\mathcal{S}^k\mathcal{L}\mathcal{S}^j$. Since by Lemma~\ref{B-lemma}, $\mathcal{S}\mathcal{S}\mathcal{S}\mathcal{S}\mathcal{L}\mathcal{S}\mathcal{S}$ cannot be a factor of $y\in\mathcal{M}$, we have $k\le 3$ or $j\le 1$; thus we can again let $t_1=y$, and we are again done.

Suppose from now on, that $n\ge 2$, and write $y=(\prod_{i=1}^n \mathcal{S}^{m_i}\mathcal{L})\mathcal{S}^{m_{n+1}}$, where each $m_i\ge 0$. For $1\le i\le n-1$, word $\mathcal{L}\mathcal{S}^{m_{i+1}}\mathcal{L}$ has one of $\mathcal{L}\mathcal{L}$, $\mathcal{L}\mathcal{S}\mathcal{L}$ or $\mathcal{L}\mathcal{S}\mathcal{S}$ as a prefix, depending on whether $m_{i+1}=0,1$ or $m_{i+1}\ge 2$, respectively. This implies that for $1\le i\le n-1$, we have $m_i\le 3$, since by Lemma~\ref{B-lemma}, no word of $\mathcal{S}^4(\mathcal{L}\mathcal{L}+\mathcal{L}\mathcal{S}\mathcal{L}+\mathcal{L}\mathcal{S}\mathcal{S})$ can be a factor of $y\in\mathcal{M}.$ For $2\le i\le n-1$, we have $m_i\le 1$, since no word of $\mathcal{L}(\mathcal{S}^2+\mathcal{S}^3)(\mathcal{L}\mathcal{L}+\mathcal{L}\mathcal{S}\mathcal{L}+\mathcal{L}\mathcal{S}\mathcal{S})$, can appear in $y$. Since $\mathcal{S}^4\mathcal{L}\mathcal{S}^2$ cannot be a factor of $y\in \mathcal{M}$, if $m_{n+1}\ge 2$, then $m_n\le 3$. We have thus established that $$y\in (\epsilon+\mathcal{S}+\mathcal{S}^2+\mathcal{S}^3)\mathcal{L}\left((\epsilon+\mathcal{S})\mathcal{L}\right)^*
\left((\epsilon+\mathcal{S}+\mathcal{S}^2+\mathcal{S}^3)\mathcal{L}\mathcal{S}\mathcal{S}\mathcal{S}^*
+\mathcal{S}^*\mathcal{L}(\epsilon+\mathcal{S})\right)
$$

Write $y=p'y't_1$, where $$p'\in(\epsilon+\mathcal{S}+\mathcal{S}^2+\mathcal{S}^3), y'\in \mathcal{L}\left((\epsilon+\mathcal{S})\mathcal{L}\right)^*,$$ $$t_1\in(\epsilon+\mathcal{S}+\mathcal{S}^2+\mathcal{S}^3)\mathcal{L}\mathcal{S}\mathcal{S}\mathcal{S}^*+
\mathcal{S}^*\mathcal{L}(\epsilon+\mathcal{S}).$$ In particular, $\mathcal{S}\mathcal{S}$ is not a factor of $y'$.

Without loss of generality, suppose $|y|\ge 7$ and
$|y'|\ge 6$. (If $|y|\le 6$ or $|y'|\le 5$, let $p_1=p'y'$, $y_1=s_1=\epsilon$, and the lemma holds. Write $y'=p^{\prime\prime}y^{\prime\prime}s_1$, where $|p^{\prime\prime}|=4$,  $|s_1|=2$. We next consider the placement in $y$, $y'$, $y^{\prime\prime}$ of hypothetical factors $\mathcal{L}^k$, $k\ge 3$:
\begin{itemize}

\item{$\mathcal{L}^k$, $k\ge 6$, cannot be a factor of $y$:} If $\mathcal{L}^6$ is a factor of $y$, so is one of $\mathcal{S}\mathcal{L}^6$, $\mathcal{L}^6\mathcal{S}$ or $\mathcal{L}^7$, since $|y|\ge 7$; this is impossible.

\item{$\mathcal{L}^5$ can only appear in $y$ as a prefix or suffix:} Otherwise, $y$ contains some two-sided extension of $\mathcal{L}^5$. As $\mathcal{L}^6$ is not a factor of $y$, this must be $\mathcal{S}\mathcal{L}^5\mathcal{S}$. This is impossible by Lemma~\ref{B-lemma}.

\item{$\mathcal{L}^4$ is not a factor of $\rho y^{\prime\prime}\sigma$, where $\rho$ is the last letter of $p^{\prime\prime}$ and $\sigma$ is the first letter of $s_1$:} The length 5 left extension of an occurrence of $\mathcal{L}^4$ in $\rho y^{\prime\prime}\sigma$ cannot be $\mathcal{L}^5$ because of the previous paragraph; it must be $\mathcal{S}\mathcal{L}^4$. Since $\mathcal{S}\mathcal{S}$ is not a factor of $y'$, the further left extension $\mathcal{L}\mathcal{S}\mathcal{L}^4$ must thus also be a factor of $y'$.  However, this forces $y'$ to contain one of the further left extensions $\mathcal{L}\mathcal{L}\mathcal{S}\mathcal{L}^4$ and $\mathcal{S}\mathcal{L}\mathcal{S}\mathcal{L}^4$, which is impossible. 

\item{$\mathcal{L}^3$ is not a factor of $y^{\prime\prime}$:} Suppose that $\mathcal{L}^3$ is a factor of $y^{\prime\prime}$. By the previous paragraph, its extension $\mathcal{S}\mathcal{L}^3\mathcal{S}$ is a factor of $\rho y^{\prime\prime}\sigma$. Since $\mathcal{S}\mathcal{S}$ is not a factor of $y'$, the extension of $\mathcal{S}\mathcal{L}^3\mathcal{S}$ to $\mathcal{L}S\mathcal{L}^3\mathcal{S}$ must be a factor of $y'$. One of the further left extensions $\mathcal{L}\mathcal{L}\mathcal{S}\mathcal{L}^3\mathcal{S}$ and $\mathcal{S}\mathcal{L}\mathcal{S}\mathcal{L}^3\mathcal{S}$ must thus occur in $y'$, but these are impossible by Lemma~\ref{B-lemma}.
\end{itemize}

We have now shown that neither of $\mathcal{S}^2$ and $\mathcal{L}^3$ can be a factor of $y^{\prime\prime}$. Thus 
$$y^{\prime\prime}\in(\mathcal{L}+\mathcal{L}\mathcal{L})(\mathcal{S}\mathcal{L}+\mathcal{S}\mathcal{L}\mathcal{L})^*.$$ Let $p^{\prime\prime\prime}$ be the longest prefix of $y^{\prime\prime}$ of the form $\mathcal{L}^k$, and write $y^{\prime\prime}=p^{\prime\prime\prime}y_1$. Letting $p_1=p'p^{\prime\prime}p^{\prime\prime\prime}$, we have $|p_1|\le 3+4+2$, so the lemma holds.
\end{proof}

\section{Parsing words of $\mathcal{M}$ using $\Phi^2$}

\begin{lemma}\label{XY}
Let $y_1\in\{\mathcal{L},\mathcal{S}\}^*$, such that $\Phi(y_1)\in\mathcal{M}$. Then $y_1$ can be written
\[
y_1 = p_2 \Phi(y_2)s_2t_2,
\]
$$\mbox{where }|p_2|,|s_2| \leq 4, y_2 \in \{\mathcal{L},\mathcal{S}\}^*\mbox{ and }$$
$$t_2 \in
\left((\epsilon+\mathcal{L}+\mathcal{L}^2+\mathcal{L}^3)
\mathcal{S}\mathcal{L}^*+\mathcal{L}^*(\epsilon+\mathcal{S}+\mathcal{S}\mathcal{L})\right)(\epsilon+\mathcal{S}+\mathcal{L}).$$
\end{lemma}
\begin{proof}
From Lemma~\ref{B-lemma}, no word of

\begin{eqnarray*}
&&(\mathcal{S}+\mathcal{L})\mathcal{S}\mathcal{S}(\mathcal{S}+\mathcal{L})
\cup(\mathcal{S}+\mathcal{L})
\mathcal{L}\mathcal{L}\mathcal{L}\mathcal{S}\mathcal{L}(\mathcal{L}+\mathcal{S}\mathcal{S}+\mathcal{S}\mathcal{L})\\
&&
\cup\Phi(\mathcal{L}\mathcal{L}\mathcal{L}(\mathcal{S}+\mathcal{L}))
\cup\Phi((\mathcal{S}+\mathcal{L})\mathcal{L}\mathcal{S}\mathcal{S}(\mathcal{S}+\mathcal{L}))
\cup\Phi((\mathcal{S}+\mathcal{L})\mathcal{S}\mathcal{S}\mathcal{S}\mathcal{S}\mathcal{S}(\mathcal{S}+\mathcal{L}))
\end{eqnarray*}
can appear in $y_1$.
This includes all length 4 two-sided extensions of $\mathcal{S}\mathcal{S}$; it follows that $\mathcal{S}\mathcal{S}$ can only appear in  $y_1$ as a prefix or suffix.

If $|y_1|\le 1$, we are done. In this case, let $p_2=y_1$, $y_2=s_2=t_2=\epsilon$. Therefore, we will assume that  $|y_1|\ge 2$, and write $y_1=p'y's'$, $|p'|=|s'|=1$. Then $\mathcal{S}\mathcal{S}$ is not a factor of $y'$.

Suppose that $|y'|_\mathcal{S}=n$. If $n=0$, the lemma is true, letting $p_2=p'$, $y_2=s_2=\epsilon$, $ t_2=y's'$. If $n=1$, write $y'=\mathcal{L}^k\mathcal{S}\mathcal{L}^j$. Since $\mathcal{L}^4\mathcal{S}\mathcal{L}^2$ is not a factor of $y_1$, $k\le 3$ or $j\le 1$; thus we can let $p_2=p'$, $t_2=y's'$, and we are again done.

Suppose from now on, that $n\ge 2$, and write $y'=(\prod_{i=1}^n \mathcal{L}^{m_i}\mathcal{S})\mathcal{L}^{m_{n+1}}$, where each $m_i\ge 0$. For $1\le i\le n-1$, $m_{i+1}\le 1$, since $\mathcal{S}\mathcal{S}$ is not a factor of $y'$. It follows that for $1\le i\le n-2$ $\mathcal{S}\mathcal{L}^{m_{i+1}}\mathcal{S}\mathcal{L}^{m_{i+2}}$ has one of $\mathcal{S}\mathcal{L}\mathcal{S}\mathcal{L}$ or $\mathcal{S}\mathcal{L}\mathcal{L}$ as a prefix. This implies that for $1\le i\le n-2$, we have $m_i\le 3$, since $\mathcal{L}^4\mathcal{S}\mathcal{L}\mathcal{S}\mathcal{L}$ and $\mathcal{L}^4\mathcal{S}\mathcal{L}\mathcal{L}$ are not factors of $y_1$. In fact,  for $2\le i\le n-2$, we have $m_i\le 2$, since $\mathcal{S}\mathcal{L}^3\mathcal{S}\mathcal{L}\mathcal{S}\mathcal{L}$ and $\mathcal{S}\mathcal{L}^3\mathcal{S}\mathcal{L}\mathcal{L}$ are not factors of $y_1$. We have thus established that $$y'\in (\epsilon+\mathcal{L}+\mathcal{L}^2+\mathcal{L}^3)\left(\mathcal{S}\mathcal{L}+\mathcal{S}\mathcal{L}\mathcal{L}\right)^*\mathcal{S}\mathcal{L}^j\mathcal{S}\mathcal{L}^k
$$
Since $\mathcal{L}^4\mathcal{S}\mathcal{L}^2$ is not a factor of $y_1$, we require $k\le 3$ or $j\le 1$.
Write $y'=p^{\prime\prime}y_2\mathcal{S}t^{\prime\prime}$ where
$p^{\prime\prime}\in (\epsilon+\mathcal{L}+\mathcal{L}^2+\mathcal{L}^3)$, $y_2\in\left(\mathcal{S}\mathcal{L}+\mathcal{S}\mathcal{L}\mathcal{L}\right)^*$, $t^{\prime\prime}\in \mathcal{S}\mathcal{L}^k\mathcal{S}\mathcal{L}^j$, $k\le 3$ or $j\le 1$.
Let $p_2=p'p^{\prime\prime}$, $s_2=\mathcal{S}$, $t_2=t^{\prime\prime}s'$. The lemma is established
.
\end{proof}
\section{Parsing words of $\mathcal{M}$ using $\Phi^3$}

\begin{lemma}\label{AB}
Let $y_2\in\{\mathcal{L},\mathcal{S}\}^*$ such that $\Phi^2(y_2)\in\mathcal{M}$. Then  $y_2$ can be written
\[
y_2 = p_3 \Phi(y_3)s_3,
\]
$$\mbox{where }|p_3|,|s_3| \leq 6, y_3 \in \{\mathcal{L},\mathcal{S}\}^*.$$
\end{lemma}
\begin{proof}
From Lemma~\ref{B-lemma}, no word of

$$\mathcal{L}\mathcal{L}\mathcal{L}(\mathcal{S}+\mathcal{L})
\cup(\mathcal{S}+\mathcal{L})\mathcal{L}\mathcal{S}\mathcal{S}(\mathcal{S}+\mathcal{L})
\cup(\mathcal{S}+\mathcal{L})\mathcal{S}\mathcal{S}\mathcal{S}\mathcal{S}\mathcal{S}(\mathcal{S}+\mathcal{L})
$$
can appear in $y_2$. These include both of the length 4 right extensions of $\mathcal{L}\mathcal{L}\mathcal{L}$; it follows that $\mathcal{L}\mathcal{L}\mathcal{L}$ can only appear in  $y_2$ as a suffix. They also include all of the length 5 two-sided extensions of $\mathcal{L}\mathcal{S}\mathcal{S}$; Thus $\mathcal{L}\mathcal{S}\mathcal{S}$ can appear in $y_2$ only as a prefix or suffix. Finally, they include all length 7 two-sided extensions of $\mathcal{S}^5$. Thus, $\mathcal{S}^5$ can only appear in $y_2$ as a suffix or prefix. 
If $|y_2|\le 4$, we are done. Assume that  $|y_2|\ge 5$, and write $y_2=p'y's'$, $|p'|=4$, $|s'|=1$. Then $\mathcal{L}\mathcal{L}\mathcal{L}$ is not a factor of $y'$. We also claim that $\mathcal{S}\mathcal{S}$ is not a factor of $y'$. Otherwise, $y_2$ has a factor $\rho \mathcal{S}\mathcal{S}$ which is not a suffix, with $|\rho|=4$. However, the length 5 suffix of $\rho \mathcal{S}\mathcal{S}$ is not a prefix or suffix of $y_2$, and contains either $\mathcal{S}^5$ or $\mathcal{L}\mathcal{S}\mathcal{S}$ as a factor; this is impossible.

Since neither of $\mathcal{L}^3$ or $\mathcal{S}^2$ is a factor of $y_2$, we have $y'\in(\epsilon+\mathcal{L}+\mathcal{L}^2)(\mathcal{S}\mathcal{L}+\mathcal{S}\mathcal{L}\mathcal{L})^*(\epsilon+\mathcal{S})$, and can write $y'=\mathcal{L}^k\Phi(y_3)\mathcal{S}^j$ where $k\le 2$, $s\le 1$. The lemma therefore holds.
\end{proof}

\section{A hierarchy of $S$'s and $L$'s}\label{hierarchy}

Combining Lemmas~\ref{SL} through \ref{AB} gives the following:

\begin{lemma}\label{Master lemma}
Let $y\in \{\mathcal{L},\mathcal{S}\}^*\cap\mathcal{M}$. Then  $y$ can be written
\[
y = p_1 \Phi(p_2 \Phi(p_3 \Phi(y_3)s_3)s_2t_2) s_1t_1,
\]
where $|p_1|,|s_1| \leq 9, |p_2|,|s_2| \leq 4, |p_3|,|s_3| \leq 6$, 
and
$$t_1 \in
(\epsilon+\mathcal{S}+\mathcal{S}^2+\mathcal{S}^3)\mathcal{L}\mathcal{S}^*+
\mathcal{S}^*(\epsilon+\mathcal{L}+\mathcal{L}\mathcal{S}),$$
$$t_2 \in
\left((\epsilon+\mathcal{L}+\mathcal{L}^2+\mathcal{L}^3)\mathcal{S}\mathcal{L}^*+
\mathcal{L}^*(\epsilon+\mathcal{S}+\mathcal{S}\mathcal{L})\right)(\epsilon+\mathcal{S}+\mathcal{L}).$$
\end{lemma}

\begin{corollary}\label{cor}
Let $y\in \{\mathcal{L},\mathcal{S}\}^*\cap\mathcal{M}$. Then there is a constant $\kappa$ such that $y$ can be written
\[
y = \pi\Phi^3(y_3)\sigma,
\]
where $\sigma$ can be written $\sigma_1\Phi(\mathcal{L}^j)\sigma_2\mathcal{S}^k\sigma_3$, with $|\pi \sigma_1\sigma_2\sigma_3| \leq \kappa.$
\end{corollary}

\begin{lemma}\label{suitability}Suppose that $\langle\mathcal{S},\mathcal{L}\rangle$ is suitable, and $|h(\mathcal{S}|$ is odd, $|h(\mathcal{L}|$ even. Let $$\Sigma=(\mathcal{S}\mathcal{L}\mathcal{S}\mathcal{L})^{-1}\Phi^3
(\mathcal{S})\mathcal{S}\mathcal{L}\mathcal{S}\mathcal{L},\mbox{ }\Lambda=(\mathcal{S}\mathcal{L}\mathcal{S}\mathcal{L})^{-1}\Phi^3
(\mathcal{L})\mathcal{S}\mathcal{L}\mathcal{S}\mathcal{L}.$$ Then $\langle\Sigma,\Lambda\rangle$ is suitable, and $|h(\Sigma)|$ is odd, $|h(\Lambda)|$ even.
\end{lemma}
\begin{proof}
Each of $|\Sigma|$, $|\Lambda|$ is odd.
 Let $$\hat{\ell}=h(\mathcal{L}\mathcal{S}\mathcal{L}\mathcal{S}
\mathcal{L}\mathcal{L}\mathcal{S}\mathcal{L}\mathcal{L}\mathcal{S})
\ell,\hspace{.1in}\hat{\mu}=\ell^R\overline{h(\mathcal{S}\mathcal{L})},\hspace{.1in}\hat{p}=\overline{\hat{\ell}^R}h(\mathcal{L}
\mathcal{S}\mathcal{L}\mathcal{S}\mathcal{L})$$

\begin{eqnarray*}
h(\Sigma)&=&h((\mathcal{S}\mathcal{L}\mathcal{S}\mathcal{L})^{-1}\Phi^3
(\mathcal{S})\mathcal{S}\mathcal{L}\mathcal{S}\mathcal{L})\\
&=&h((\mathcal{S}\mathcal{L}\mathcal{S}\mathcal{L})^{-1}
\mathcal{S}\mathcal{L}\mathcal{S}\mathcal{L}\mathcal{L}
\mathcal{S}\mathcal{L}\mathcal{S}\mathcal{L}\mathcal{L}\mathcal{S}\mathcal{L}\mathcal{L}
\mathcal{S}\mathcal{L}\mathcal{S}\mathcal{L})\\
&=&
h(\mathcal{L}\mathcal{S}\mathcal{L}\mathcal{S}
\mathcal{L}\mathcal{L}\mathcal{S}\mathcal{L}\mathcal{L}\mathcal{S}
\mathcal{L}\mathcal{S}\mathcal{L})\\
&=&
h(\mathcal{L}\mathcal{S}\mathcal{L}\mathcal{S}
\mathcal{L}\mathcal{L}\mathcal{S}\mathcal{L}\mathcal{L}\mathcal{S})
\ell\ell^R
\overline{h(\mathcal{S}\mathcal{L})}\\
&=&\hat{\ell}\hat{\mu}
\end{eqnarray*}
For a word $z\in\{\mathcal{S},\mathcal{L}\}^*$ with $|z|$ even, we observe that $\overline{h(z^R)}=(h(z))^R$.  Therefore, we also have
\begin{eqnarray*}
\Sigma&=&
h(\mathcal{L}\mathcal{S}\mathcal{L}\mathcal{S}
\mathcal{L}\mathcal{L}\mathcal{S}\mathcal{L}\mathcal{L}\mathcal{S}
\mathcal{L}\mathcal{S}\mathcal{L})\\
&=&h(\mathcal{L}\mathcal{S})h(\mathcal{L})\overline{h(\mathcal{S}
\mathcal{L}\mathcal{L}\mathcal{S}\mathcal{L}\mathcal{L}\mathcal{S}
\mathcal{L}\mathcal{S}\mathcal{L})}\\
&=&
(\overline{h(\mathcal{S}\mathcal{L})})^R\ell\ell^R\left(h(\mathcal{L}\mathcal{S}
\mathcal{L}\mathcal{S}\mathcal{L}\mathcal{L}\mathcal{S}\mathcal{L}
\mathcal{L}\mathcal{S})\right)^R\\
&=&\hat{\mu}^R\hat{\ell}^R
\end{eqnarray*}

Further,

\begin{eqnarray*}
h(\Lambda)&=&h((\mathcal{S}\mathcal{L}\mathcal{S}\mathcal{L})^{-1}\Phi^3
(\mathcal{L})\mathcal{S}\mathcal{L}\mathcal{S}\mathcal{L})\\
&=&h((\mathcal{S}\mathcal{L}\mathcal{S}\mathcal{L})^{-1}
\mathcal{S}\mathcal{L}\mathcal{S}\mathcal{L}\mathcal{L}
\mathcal{S}\mathcal{L}\mathcal{S}\mathcal{L}\mathcal{L}\mathcal{S}\mathcal{L}\mathcal{L}
\mathcal{S}\mathcal{L}\mathcal{S}\mathcal{L}\mathcal{L}\mathcal{S}\mathcal{L}\mathcal{L}
\mathcal{S}\mathcal{L}\mathcal{S}\mathcal{L})\\
&=&h(
\mathcal{L}\mathcal{S}\mathcal{L}\mathcal{S}\mathcal{L}\mathcal{L}\mathcal{S}\mathcal{L}\mathcal{L}
\mathcal{S}
\mathcal{L}
\mathcal{S}\mathcal{L}\mathcal{L}\mathcal{S}\mathcal{L}\mathcal{L}\mathcal{S}\mathcal{L}\mathcal{S}
\mathcal{L})\\
&=&h(\mathcal{L}\mathcal{S}\mathcal{L}\mathcal{S}
\mathcal{L}\mathcal{L}\mathcal{S}\mathcal{L}\mathcal{L}\mathcal{S})
h(\mathcal{L})\overline{h(\mathcal{S}\mathcal{L}\mathcal{L}\mathcal{S}\mathcal{L}\mathcal{L}
\mathcal{S}\mathcal{L}\mathcal{S}\mathcal{L})}\\
&=&h(\mathcal{L}\mathcal{S}\mathcal{L}\mathcal{S}
\mathcal{L}\mathcal{L}\mathcal{S}\mathcal{L}\mathcal{L}\mathcal{S})
\ell\hspace{.1in}\ell^R
\left(h(\mathcal{L}\mathcal{S}\mathcal{L}\mathcal{S}
\mathcal{L}\mathcal{L}\mathcal{S}\mathcal{L}\mathcal{L}\mathcal{S})\right)^R\\
&=&\hat{\ell}\hat{\ell}^R
\end{eqnarray*}
Finally,
\begin{eqnarray*}
h(\Lambda)
&=&h(\mathcal{L}\mathcal{S}\mathcal{L}\mathcal{S}
\mathcal{L}\mathcal{L}\mathcal{S}\mathcal{L}\mathcal{L}\mathcal{S})
\ell\hspace{.1in}\ell^R
\overline{h(\mathcal{S}\mathcal{L}\mathcal{L}\mathcal{S}\mathcal{L}\mathcal{L}
\mathcal{S}\mathcal{L}\mathcal{S}\mathcal{L})}\\
&=&h(\mathcal{L}\mathcal{S}\mathcal{L}\mathcal{S}
\mathcal{L}\mathcal{L}\mathcal{S}\mathcal{L}\mathcal{L}\mathcal{S})
\ell\hspace{.1in}\ell^R
\overline{h(\mathcal{S}\mathcal{L})}\hspace{.05in}\overline{h(\mathcal{L}\mathcal{S})}\hspace{.05in}\overline{h(\mathcal{L})}h(\mathcal{L}
\mathcal{S}\mathcal{L}\mathcal{S}\mathcal{L})\\
&=&h(\mathcal{L}\mathcal{S}\mathcal{L}\mathcal{S}
\mathcal{L}\mathcal{L}\mathcal{S}\mathcal{L}\mathcal{L}\mathcal{S})
\ell\hspace{.1in}\ell^R
\overline{h(\mathcal{S}\mathcal{L})}\hspace{.05in}\overline{h(\mathcal{L}\mathcal{S})}
\hspace{.05in}\overline{\hat{\ell}}\hspace{.05in}\overline{\hat{\ell}^R}h(\mathcal{L}
\mathcal{S}\mathcal{L}\mathcal{S}\mathcal{L})\\
&=&\hat{\ell}\hat{\mu}\overline{\hat{\mu}^R}\hat{p}.
\end{eqnarray*}
\end{proof}

This result combines with Corollary~\ref{cor} to allow us to parse words of $\mathcal{M}$. Let $L_0=L$, $S_0=S$. Supposing that $\langle S_i,L_i\rangle$ is suitable, let $\mathcal{L}=L_i$, $\mathcal{S}=S_i$, and 
$$L_{i+1}=(S_iL_iS_iL_i)^{-1}\Phi^3(L_i)S_iL_iS_iL_i,\hspace{.1in}S_{i+1}=(S_iL_iS_iL_i)^{-1}\Phi^3(L_i)S_iL_iS_iL_i.$$
Since $\langle S,L\rangle$ is suitable, all of the pairs $\langle S_i,L_i\rangle$ will be suitable by Lemma~\ref{suitability}. Suppose $y\in\{S,L\}^*\cap\mathcal{M}$. By repeatedly applying Corollary~\ref{cor},
we write $y=\hat{\pi}\upsilon\hat{\sigma}$ where $\upsilon\in\{S_i,L_i\}^*$.

\section{Upper bound on growth rate}

Define
\[
\mathcal{N} = \{z \in \{0,1\}^* : z \text{ avoids } xxx^R\}.
\]

\begin{theorem}\label{poly_bin}
The number of words in $\mathcal{N}$ of length $n$ is $O(n^{\lg n+c})$, some constant $c$.
\end{theorem}

To prove this theorem, it suffices to show that the number of words in $\mathcal{K}$ of length $n$ is $O(n^{\lg n+c})$, some constant $c$.

From Theorem 1, it suffices to prove the following:

\begin{theorem}\label{poly_SL}
The number of words in $\mathcal{M}$ of length $n$ is $O(n^{\lg n+c})$, some constant $c$.
\end{theorem}

\begin{proof}[Proof of Theorem~\ref{poly_SL}]
 Let $y \in \mathcal{M}$ have length $n$. Choose $\langle \mathcal{S},\mathcal{L}\rangle=\langle S,L\rangle$. Then iteration of
Corollary~\ref{cor} gives
\[
y = p_1 \Phi^3( p_2 \Phi^3( p_3 \cdots p_m \Phi^3(\epsilon) s_m
\cdots s_3) s_2 ) s_1,
\]
where $m \leq (\lg n)/3$.  For $i \in \{1,\cdots,m\}$ we have
\[
s_i = 
\sigma_{1,i}\Phi(\mathcal{L}^{j_i})\sigma_{2,i}\mathcal{S}^{k_i}\sigma_{3,i}
\]
Since $|p_i\sigma_{3,i}\sigma_{2,i}\sigma_{1,i}| \leq \kappa$, there is a
constant $\alpha$ such that there are at most $\alpha$ choices for $(p_i,
\sigma_{i,3},\sigma_{i,2},\sigma_{i,1})$. This gives a number of choices for $\{ (p_i,
\sigma_{i,3},\sigma_{i,2},\sigma_{i,1})\}_{i=1}^m$ which is polynomial in $n$.

This leaves the problem of 
bounding the number of choices of the $j_i$ and $k_i$. 

We have
\begin{eqnarray*}
n&\ge& |\Phi^3(\Phi^3(\cdots \Phi^3(\epsilon) \Phi(\mathcal{L}^{j_m})\mathcal{S}^{k_m}
\cdots \Phi(\mathcal{L}^{j_3})\mathcal{S}^{k_3}) \Phi(\mathcal{L}^{j_2})\mathcal{S}^{k_2} ) \Phi(\mathcal{L}^{j_1})\mathcal{S}^{k_1}|\\
&=&\sum_{i=1}^m\left(j_i|\Phi^{3i-2}(\mathcal{L})|+k_i|\Phi^{3i-3}(\mathcal{S})|\right)\\
&=&\sum_{i=1}^m\left(j_i\mathcal{F}_{6i-3}+k_i\mathcal{F}_{6i-6}\right)
\end{eqnarray*}

It follows that the number of choices for the $j_i$, $k_i$ is less
than or equal to the number of partitions (with repetition) of $n$ with parts chosen
from $\{\mathcal{F}_{3i}\}_{i=0}^\infty$. Since $\mathcal{F}_{3i}\ge 2^i$, this is less than or equal to the number of partitions of $n$ into 
powers of 2. Mahler \cite{mahler} showed that the number $p(n,r)$ of partitions of $n$ into powers of $r$ satisfies
$$\lg p(n,r)\sim \frac{\lg^2 n}{\lg^2 r};$$  
thus, $p(n,2)\sim Cn^{\lg n}$ where $C$ is constant.
The result follows.\end{proof}
\section{Lower bound on growth}
Let $\psi:\{L,S\}^*\rightarrow\{L,S\}^*$ be given by
 $$\psi(S)=LSL,\psi(L)=LSLSL.$$
Since $\psi(S)$, $\psi(L)$ are palindromes, we have 
$$\psi(u^R)=(\psi(u))^R, u\in\{L,S\}^*.$$
Letting $\langle \mathcal{S},\mathcal{L}\rangle=\langle S,L\rangle$, we find that
$\psi=(\mathcal{LSL})^{-1}D^3\mathcal{LSL}.$ It follows that $|\psi^k(S)|=\mathcal{F}_{3k}$, 
$|\psi^k(L)|=\mathcal{F}_{3k+1}$

Define languages $\mathscr{L}_i$ by 
$$\mathscr{L}_0=LS^*,\mathscr{L}_{i+1}=\psi(\mathscr{L}_i)LS^*.$$ 
Let $\mathscr{L}=\cup_{i=0}^\infty\mathscr{L}_i$. 

A word $w\in\mathscr{L}$ has the form
$$w=\psi(\psi(\cdots \psi(\psi(LS^{k_m})LS^{k_{m-1}})\cdots)LS^{k_2})LS^{k_1})LS^{k_0}$$
so that the number of words of $\mathscr{L}$ of length $n$ is the number of partitions of $n$ of the form

$$n=\sum_{i=0}^{m}\left(\mathcal{F}_{3i+1}+k_i\mathcal{F}_{3i}\right).$$
Since $\mathcal{F}_{i}\le 2^i$, this is greater than or equal to 
 the number of partitions of $n$ of the form

$$n=\sum_{i=0}^{m}\left( 2^{3i+1}+k_i2^{3i}\right),$$
which is greater than or equal to 
 the number of partitions of $n$ of the form

$$n=\sum_{i=0}^{m} (k_i+1)2^{3i+1}.$$

This, in turn, is at least half of
the number of partitions of $n$ of the form

$$n=\sum_{i=0}^{m}k_i2^{3i+1},$$

which is the number of partitions of $n/2$ of the form

$$n/2=\sum_{i=0}^{m} k_i8^{i}.$$

Following Mahler \cite{mahler}, this is
$p(n/2,8)\sim Cn^{\lg n}/n^2$, where $C$ is constant.
We will show that no word of $h(\mathscr{L})$ has a non-empty factor $xxx^R$, so that this gives a lower bound on $\mathcal{N}$.

One checks the following:
\begin{lemma}
No word of $\mathscr{L}$ has any of the following factors:
$$L^3, SSL, SLSLS, LSLSLLSLSLLSLSL=\psi(L^3), LLS L LSLLS L SL$$
$$LSLLSLSLLSLLSLSLLSL=\psi(SLSLS).
$$
\end{lemma}

\begin{theorem}\label{psi(L)S*}
No word of $h(\mathscr{L})$ contains a non-empty word of the form $xxx^R$.
\end{theorem}
\begin{proof}
Suppose $w\in\mathscr{L}$, and $xxx^R$ is a non-empty factor of $h(w)$. Let $$W=\left((h(S)+h(L))(\overline{h(S)}+\overline{h(L)})\right)^*=\left((00100+00100100)(11011+11011011)\right)^*.$$ Thus $h(w)$ is a factor of a word of $W$. Note that none of $000$, $111$, $0101$, $1010$, $001011$, $110011$, $010010010$, is a factor of any word of $W$, nor thus, of $w$. Also, $\ell=0010$ is always followed by 01 in any word of $W$, while $\overline{\ell^R}=1011$ is always preceded by 01.

If $|x|\le 2$, then $h(w)$ contains a factor 000, 111, 010110 or 101011. The last two contain 0101, so this is impossible. Assume therefore that $|x|\ge 3$ and write $x = x'\alpha\beta\gamma$, where
$\alpha$, $\beta$, $\gamma\in\{0,1\}$. Then $\alpha\beta\gamma\gamma\beta\alpha$ is a factor of $xxx^R$. Suppose that $\gamma=0$. (The other case is similar.) Since 000 is not factor of $w$, we can assume that $\beta=1$. Since $110011$ is not a factor of $w$, $\alpha\beta\gamma=010$. If $|x|=3$, then $xxx^R$ is $010010010$, which is not a factor of $w$. We conclude that $|x|\ge 4$. Since $1010$ is not a factor of $w$, $\ell=0010$ is a suffix of $x$. Write $x=x^{\prime\prime}\ell,$ so that 
                         $$xxx^R=x^{\prime\prime}\ell x^{\prime\prime}\ell\ell^R(x^{\prime\prime})^R=
x^{\prime\prime}\ell x^{\prime\prime}h(L)(x^{\prime\prime})^R.$$
Since $x^{\prime\prime}\ell x^{\prime\prime}$ precedes $h(L)$ in a word of $W$, the length 4 suffix of $
x^{\prime\prime}\ell x^{\prime\prime}$ must be $1011$; since $x^{\prime\prime}$ follows $\ell$ in $h(w)$, it follows that $x^{\prime\prime}$ begins with $0$. Therefore, $|x^{\prime\prime}|\ge 5$. It follows that $x^{\prime\prime}$ must end with 11011, so that, in fact, $|x^{\prime\prime}|\ge 6$, and 011011 is a suffix of $x^{\prime\prime}$. If $|x^{\prime\prime}|=6$, then 
                         $$xxx^R=0110110010\hspace{.05in} 0110110010 \hspace{.05in} 0100110110=011011h(SSL)110110.$$
This forces $SSL$ to be a factor of $w$, which is impossible, since $w\in\mathscr{L}$. Thus $|x^{\prime\prime}|\ge 7.$

Since 0101 is not a factor of $w$, if suffix 011011 of $x^{\prime\prime}$ is preceded by 1, it is preceded by 11, and $\overline{h(L)}\ell$ is a suffix of $x$. This forces $xx^R$ to have 
$$\overline{h(L)}\ell\ell^R\overline{h(L)^R}=\overline{h(L)}h(L)\overline{h(L)}$$
as a factor, forcing $LLL$ to be a factor of $w$, which is impossible. We conclude that 0011011 is a suffix of $x^{\prime\prime}$. Since $x^{\prime\prime}$ follows $\ell$ in $w$, 01 must be a prefix of $x^{\prime\prime}$. Suppose $011$ is a prefix of $x^{\prime\prime}$. Since 0011011 is a suffix, 
then $x^{\prime\prime}\ell x^{\prime\prime}$ has factor $$0011011\ell 011=00\overline{h(S)}h(S)11$$ and $w$ has a factor $SSuL$ for some $u$; this is impossible. We conclude that $010$ is a prefix of $x^{\prime\prime}$; since $0101$ is not a factor of $w$, in fact, $0100=\ell^R$ is a prefix of $x^{\prime\prime}$. In total,
$$xxx^R=\ell^R\hat{x}\ell\hspace{.05in}\ell^R\hat{x}\ell \hspace{.05in}\ell^R\hat{x}^R\ell$$
The `bracketing' by $\ell$ and $\ell^R$ forces $w$ to contain a factor $uLuLu^R$, where $|u|$ is odd. 

Consider the shortest factor $uLuLu^R$ or $w$, where $|u|$ is odd.

If the last letter of $u$ is $L$, then $LLL$ is a central factor of $uLu$. This is impossible. Thus $S$ is a suffix of $u$. If $u=S$, then $uLuLu^R=SLSLS$, which is not a factor of any word of $\mathscr{L}$. We conclude that $|u|>1$, so that $|u|\ge 3$, since $|u|$ is odd.

Since $SSL$ is not a factor of $w$, the length 3 suffix of $uL$ is $LSL$. This makes $LSLSL$ a central factor of $uLu^R$. Since $SLSLS$ is not a factor of $w$, the length 3 suffix of $u$ is $LLS$. If $u=LLS$, then $Lu$ has prefix $LLL$,  which is not a factor of $w$. We conclude that $|u|\ge 5$.

Since neither of $LLL$ and $SS$ is a factor of $w$, we conclude that $LSLLS$ is the length 5 suffix of $u$. If $u=LSLLS$, then $uLuLu^R=LSLLS L LSLLS L SLLSL$, with illegal factor  $LLS L LSLLS L SL$. Thus $|u|\ge 7$.

If the length 7 suffix of $u$ is $LSLSLLS$, then a central factor $uLu^R$ is $LSLSLLSLSLLSLSL$, which is not a factor of $w$. We conclude that the length 7 suffix is $SLLSLLS$.

Write 
$w=\psi(v)LS^k$ for some $v\in\mathscr{L}$, some $k\ge 0$. Since $|w|_L>1$,  $v\ne\epsilon$.
Then $w$ has suffix $LLS^k$, and prefix $uLuLSL$ of $uLuLu^R$ must be a factor of $\psi(v)$. Let $L(LSL)^mL$ be a factor of $uLu$ where $m$ is as large as possible. Since $uLuLSL$ has suffix $LSLSL$, and  $uLuLSL$ is a factor of $\psi(v)$, word $L(LSL)^mLSLSL$ must be a factor of $uLuLSL$. If $m\ge 2$, then $uLuLu^R$ has illegal factor
$LLSLLSLLSLSL$. We conclude that $m=1$, so that $LLSLLSLL$ is not a factor of $uLu$

In the context of $uLu$, word $u$ follows the suffix $LLSLLSL$ of $uL$. Therefore, $u$ cannot have $L$  as a prefix or $uLu$ contains the factor $LLSLLSLL.$ It follows that $SL$ is a prefix of $u$. However, a prefix of $u$ cannot be $SLS$; otherwise $uLu$ would have factor $uLSLS$ which has illegal suffix $SLSLS$. It follows that the length 3 prefix of $u$ is $SLL$.

Write $$u=SL:Lu'SL:LSL:LS$$
The colons indicate boundaries in $u$ between instances of $\psi(S)$ and $\psi(L)$. Thus, we may write
$u=SL\psi(u^{\prime\prime})LS,$
for some word $u^{\prime\prime}$ in $\mathscr{L}$. Since $|\psi(S)|\equiv|\psi(L)|\equiv 1\mbox{ (modulo }2),$ we have
$$|u|\equiv|\psi(u^{\prime\prime})|\equiv|u^{\prime\prime}|\mbox{ (modulo }2).$$
Then 
\begin{eqnarray*}
uLuLu^R&=&SL\psi(u^{\prime\prime})LSLSL\psi(u^{\prime\prime})LSLSL(\psi(u^{\prime\prime}))^RLS\\
&=&SL\psi(u^{\prime\prime}Lu^{\prime\prime}L(u^{\prime\prime})^R)LS.
\end{eqnarray*}

Recall that $w=\psi(v)LS^k$. Although the suffix $LS$ of $uLuLu^R$ may  occur here as a prefix of $LS^k$, certainly $uLuLu^R(LS)^{-1}$ is in $\psi(v).$
We conclude that $u^{\prime\prime}Lu^{\prime\prime}L(u^{\prime\prime})^R$ is a factor of $\mathscr{L}$, where $u^{\prime\prime}$ has odd length shorter than $u$. This is a contradiction.
\end{proof}

\end{document}